\documentclass[10pt]{llncs}

\usepackage[framemethod=tikz]{mdframed}

\usepackage{amsmath}
\usepackage{amsfonts}
\usepackage{amssymb}

\usepackage{listings}

\usepackage{caption}
\usepackage{subcaption}

\usepackage{enumerate}

\usepackage{tikz}
\usepackage{pgfplots}
\usetikzlibrary{patterns,fit, calc, positioning}
\pgfplotsset{width=7cm, compat=1.10}
\usepgfplotslibrary{fillbetween}

\usepackage{multirow}

\usepackage{url}
\usepackage{hyperref}
\hypersetup{
  colorlinks=true,
  linkcolor=blue,
  citecolor=green
}

\usepackage[backend=bibtex,style=alphabetic,maxnames=5,minnames=5,maxbibnames=99,firstinits=true,
doi=false,isbn=false,url=false
]{biblatex}
\usepackage[nameinlink]{cleveref}
\crefname{table}{table}{tables}
\Crefname{table}{Table}{Tables}
\crefname{figure}{figure}{figures}
\Crefname{figure}{Figure}{Figures}
\crefname{section}{section}{sections}
\Crefname{section}{Section}{Sections}
\crefname{claim}{claim}{claims}
\Crefname{claim}{Claim}{Claims}
\Crefname{question}{Question}{Questions}
\Crefname{appendix}{Appendix}{Appendices}

\title{Entropy of Independent Experiments, Revisited}

\author{Maciej Skorski}
\institute{IST Austria}

\begin{document}

\maketitle

\begin{abstract}
The weak law of large numbers implies that, under mild assumptions on the source, the Renyi entropy per produced symbol converges (in probability) towards the Shannon entropy rate. 

This paper quantifies the speed of this convergence for sources with independent (but not iid) outputs, generalizing and improving the result of  Holenstein and Renner (IEEE Trans. Inform. Theory, 2011). 
\begin{enumerate}[(a)]
\item we characterize sources with \emph{slowest convergence} (for given entropy): their outputs are mixtures of a uniform distribution and a unit mass.
\item based on the above characterization, we establish faster convergences
in \emph{high-entropy} regimes.
\end{enumerate}

We discuss how these improved bounds may be used to better quantify security of outputs of random number generators.

In turn, the characterization of "worst" distributions can be used 
to derive sharp "extremal" inequalities between Renyi and Shannon entropy.

The main technique is \emph{non-convex programming}, used to characterize distributions of possibly large exponential moments under certain entropy.

% In other words, 
%5high-entropy sources
%This corresponds to the well established practice of prefering sources with %higher-entropy rate. 
 
% that provide
% quantitative bounds for the entropy rate of sources producing independent symbols (a variant of the asymptotic equipartition property). 

%, by which
%we characterize the "slowest" source (of a certain entropy): it combines a unit mass and a uniform distribution. As a byproduct we get the same characterization in related problems of maximizing Renyi entropy.
\end{abstract}


\section{Introduction}

%One shot vs repeated experiment

%The Asymptotic Equiparition Property states that...

It is well known that under mild assumptions on the source (independent and identical outputs~\cite{Renner2004}, independent but not identical outputs~\cite{HolensteinRenner2011}, ergodicity~\cite{Schoenmakers2007}, strong converse property~\cite{Koga2013})
the rate of min-entropy (in fact, Renyi entropy of any positive order) converges in probability\footnote{For the convergence in probability we understand the entropy as conditioned over some set of probability $1-o(1)$. This is strictly related to smooth entropy~\cite{Renner2004}.}.
 towards the Shannon entropy rate. More precisely, for the source producing outputs $X_1,X_2,\ldots$ and $x\gets X^n \overset{def}{=} X_1,\ldots,X_n$, under these assumptions for $n\to\infty$ we have
\begin{align}\label{eq:convergence}
 \frac{1}{n}\log\frac{1}{P_{X^{n}}(x)} = \frac{1}{n}H(X^{n}) + o(1)
 \quad \text{w.p. } 1-o(1)
\end{align}
which can be seen a demonstration of the weak law of large numbers\footnote{By the definition of Shannon entropy we have $\mathbb{E}\log\frac{1}{P_{X^n}(\cdot)} = H(X^n)$.}. %Note that the approximation holds for sources as opposed to one-shot experiments, when the Shannon and Renyi entropy of a single outcome can be arbitrary different. 

In information theory results of this sort are often refereed to  as generalizations of the \emph{Asymptotic Equipartition Property}, because they establish 
that with overwhelming probability the sequences produced by the source
are (roughly) equally likely.

Under general assumptions, there is basically not much more to say about \eqref{eq:convergence}. This paper is concerned with quantitative bounds, which are possible when the source produces independent outputs.
Thus, we are interested in inequalities 
\begin{align}\label{eq:convergence_2}
\Pr_{x_1,\ldots,x_n\gets X_1,\ldots,X_n}\left[\left| \sum_{i=1}^{n}\log\frac{1}{P_{X_{i}}(x_i)} - \sum_{i=1}^{n} H(X_i)\right| > n\delta\right] \leqslant \epsilon
\end{align}
where $X_i$ are independent random variables with finitely many outcomes\footnote{Our results are valid when $X_i$ have different alhapbets,
however for the clarity of the presenation we later assume that
they all are over some fixed $\mathcal{X}$.}. 
Bounds of this sort find applications in cryptography,
quantifying the the conversion between Shannon entropy (more convenient to quantify) and the min-entropy (required for security), for a series of experiments. One example are theoretical constructions of pseudorandom generators~\cite{HILL99,Holenstein2006}, which use a variant of \eqref{eq:convergence_2}. Another important application is a justification of the entropy evaulation methodology for random number generators~\cite{AIS31,Barker2016NISTDS}, where best available tests quantify Shannon entropy~\cite{Maurer92a,Coron1999}.

%The simplest bound of this sort
%The fact that $\epsilon = o(1)$ when $n$ is large is merely 
%the 

Good bounds of the form \eqref{eq:convergence_2} are obtained by techniques from \emph{large deviations theory}, applied to the
random variables $Z_i = \log\frac{1}{P_{X_{i}}(\cdot)}$ called
 \emph{surprises} of $X_i$. 
Note that $Z_i$ are \emph{unbounded}, hence standard inequalities like Chernoff-Hoeffding bounds don't apply. 
The solution\footnote{There are other
approaches, for example ignoring large surprises or using the concept of typical sets, but they lead to worse bounds as discussed in~\cite{HolensteinRenner2011}.
} is to work directly with \emph{moment generating functions} of each $Z_i$. If we know upper bounds on $\mathbb{E}\exp(t Z_i)$,
where $t$ is a parameter, then \eqref{eq:convergence_2} follows
by the Markov inequality $\Pr[\sum_{i}Z_i \geqslant n\delta]\leqslant
\prod_{i}\exp\left(tZ_i\right)\cdot \exp(-tn\delta)$, optimized over $t$.
%We follow this approach, deriving bounds that explicitly depend
%on the entropy of each $X_i$.
%Our problem therefore reads as follows
In bounds \eqref{eq:convergence_2} we may also want to capture
 some information about the Shannon entropy of the source. Technically, the problem then reduces to
\begin{mdframed}[innertopmargin=2pt, leftline=true, topline = true]
\begin{quote}
\textbf{(Problem)} For any alphabet $\mathcal{X}$,
find best possible bounds on the exponential moments of the surprise of 
a distribution $X$ over $\mathcal{X}$
\begin{align*}
\underset{X}{\textrm{max}}\ \mathbb{E}_{x\sim X}\exp\left(t\log \frac{1}{P_X(x)}\right) \leqslant \quad ?
\end{align*}
assuming that $X$ has Shannon entropy $k$, for parameters $k$ and $t$.
\end{quote}
\end{mdframed}
We follow this approach and derive best possible bounds for this technique.

\subsection{Related works and our results}

\subsubsection{Related works}
Best bounds of the form \eqref{eq:convergence_2} so far were due to Holenstein and Renner~\cite{HolensteinRenner2011}. Their argument uses  calculus to derive bounds on the exponential moments of the surprise of $X$. This leads to
$\epsilon = \exp\left(-\Omega(1)\cdot \frac{n\delta^2}{\log^2|\mathcal{X}|} \right)$. However the matching (up to a constnat in the exponent) lower bound $\epsilon = \exp\left(-O(1)\cdot \frac{n\delta^2}{\log^2|\mathcal{X}|} \right)$ is known only for low or moderate entropy rates\footnote{Note that in~\cite{HolensteinRenner2011} the bounds are shown to be optimal only for moderate entropy, namely of about $\frac{1}{2}\log|\mathcal{X}|$ per sample (cf the proof of Theorem 3 in~\cite{HolensteinRenner2011}.}.

\subsubsection{Our results and techniques in a nutshell}
We provide explicit bounds in terms of entropies of $X_i$, instead of the alphabet as in~\cite{HolensteinRenner2011}. Roughly speaking, we replace 
the factor logarithmic in the alphabet by the entropy efficiency.

In particular, we obtain significant improvements 
in \emph{high-entropy regimes} where the deficiency $\Delta_i = H_0(X_i) - H(X_i)$ is relatively small, for example a fraction of the length (here $H_0(\cdot)$ is the logarithm of the support of $X_i$). As a consequence, \eqref{eq:convergence} converges faster when the distributions $X_i$ have high entropy; alternatively, we get better accuracy $\delta$ with fewer samples.
We summarize our bounds in \Cref{tab:comparison}.
\begin{table}
\captionsetup{font=small}
\centering
\renewcommand*{\arraystretch}{1.2}
\resizebox{0.99\textwidth}{!}{
\begin{tabular}{|l|r|r|l|l|}
\hline
author/reference & number of samples $n$ & regime & technique \\
\hline
\cite{HolensteinRenner2011} &  $\log^2N \cdot \delta^{-2}\log\frac{1}{\epsilon}$ & & bounds on exponential moments\\
\hline
%\Cref{cor:tail_no_constraints} 
%& $\log^2N \cdot \delta^{-2}\log\frac{1}{\epsilon}$ & & sub-gaussian tails  \\
%\hline
%\Cref{cor:tail_constraints} 
%& $\Delta\log N \cdot \delta^{-2}\log\frac{1}{\epsilon}$  & 
%$\delta < \frac{1}{2\log N}$ & 
%sub-exponential tails \\
%\hline
this paper, \Cref{cor:sub_exp_tails_improved} (a) & $\log (\Delta N)\cdot \max\left(\delta^{-1}, \Delta\delta^{-2}\right)\log\frac{1}{\epsilon}$  &  $\Delta > \frac{1}{N}$ &
\multirow{2}{*}{\parbox{5cm}{ optimized
exponential moments \\ subexpoential tails}} \\
\cline{1-3}
this paper, \Cref{cor:sub_exp_tails_improved} (b) & $\max(\Delta\cdot \delta^{-1},1)\log\frac{1}{\epsilon}$  & $\Delta <\frac{1}{N}$ \\
\hline
\end{tabular}
}
\caption{Summary of our bounds and comparison with related works. The alphabet size is $|\mathcal{X}| = N$, the number of samples $n$ is such that 
\eqref{eq:convergence_2} holds with accuracy $\delta$ and error probability $\epsilon$, and defficiencies are are bounded by $\Delta$. Note that our bounds for the low entropy setting $\Delta = \log N$ reduce to \cite{HolensteinRenner2011}.}
\label{tab:comparison}
\end{table}

The bounds are relevant to random number generation, where 
we improve the known relation between min-entropy (which is the best notion to be used in vulnerability analysis) and Shannon entropy (much easier to estimate in practice~\cite{AIS31}). We discuss this application in the next paragraph. 

To quantify the convergence in terms of entropies of $X_i$, 
we use {non-convex optimization} to find \emph{best bounds on exponential moments} of $X_i$ under \emph{given entropy}. 
While the optimization is rather non-trivial, the extreme distribution has 
a simple shape: it combines a unit mass and a uniform distribution. This analysis not only guarantees that the bounds are best possible, but
as a byproduct solves the related problem of Renyi entropy maximization under fixed Shannon entropy. This way we obtain extremal inequalities between these entropies.

\subsubsection{Applications to random number generators}
The motivation for studying \emph{high entropy regimes} comes from \emph{true random number generators}. Roughly speaking, they are devices which postprocess samples from a physical entropy source into a sequence of almost independent and unbiased bits.
\begin{enumerate}[(a)]
\item \emph{output security}: output bits are required to have very high Shannon entropy\footnote{} per bit (e.g. more than 0.997) and no dependencies~\cite{AIS31}; for example, early Intel (hardware) generators were estimated to generate about 0.999 Shannon entropy per bit~\cite{jun1999intel}. 

To illustrate our bounds, 
suppose that outputs are generated in $8$ bits chunks. We have $|\mathcal{X}| = 2^8$ states and the Shannon entropy deficiency is $\Delta = (1-0.997)\cdot 8 = 0.024$. The security (min-entropy) implied by bounds in \Cref{tab:comparison} at the confidence $1-\epsilon$ where $\epsilon=2^{-60}$ is illustrated in \Cref{eq:numerical_convergence} below.

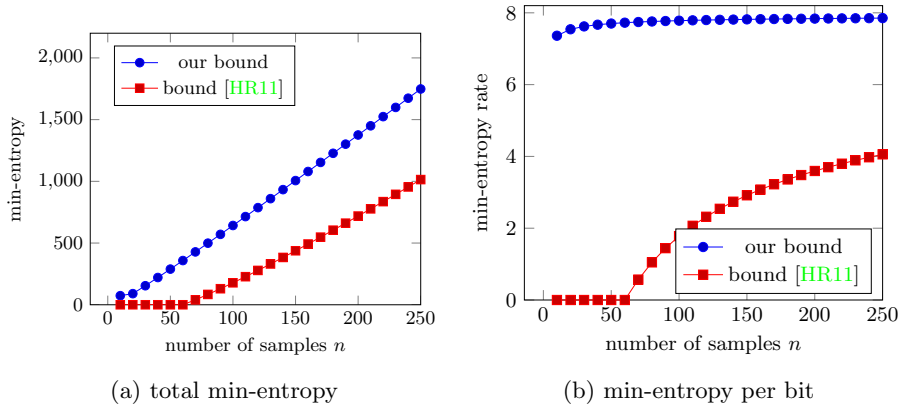
\begin{figure}[h!]
\begin{subfigure}[b]{.5\textwidth}
\resizebox{0.99\linewidth}{!}{
\begin{tikzpicture}
\begin{axis}[
    	title = { },  % whatever name you want
    	ylabel = {min-entropy},
    	xlabel = {number of samples $n$},
    	ymin = 0, ymax = 2200,
    	xmax = 250,
	legend style = {at = {(axis cs:5,2100)}, anchor = north west},
    	legend entries = {our bound, bound \cite{HolensteinRenner2011}}
    	]
    	\addplot+[blue] table {plot1.csv};
    	\addplot+[red] table {plot2.csv};
    \end{axis}
\end{tikzpicture}
}
\caption{total min-entropy}
\label{fig:1}
\end{subfigure}
\begin{subfigure}[b]{.5\textwidth}
\resizebox{0.99\linewidth}{!}{
\begin{tikzpicture}
\begin{axis}[
    	title = {},  % whatever name you want
    	ylabel = {min-entropy rate},
    	xlabel = {number of samples $n$},
    	ymin = 0, ymax = 8.2,
    	xmax = 250,
	legend style = {at = {(axis cs:242,0.17)}, anchor = south east},
    	legend entries = {our bound, bound \cite{HolensteinRenner2011}}
    	]
    	\addplot+[blue] table {plot3.csv};
    	\addplot+[red] table {plot4.csv};
    \end{axis}
\end{tikzpicture}
}
\caption{min-entropy per bit}
\label{fig:2}
\end{subfigure}
\caption{Numerical comparison of our and previous bounds, applied to 
the sequence of $n$ independent samples of very high Shannon entropy (required for the output of a true random number generator~\cite{AIS31}). Each sample is $8$-bit long, the Shannon entropy per bit is 0.997, the error probability $\epsilon=2^{-60}$. The vertical axes show the min-entropy (conditioned over a set of probability $1-\epsilon$).}
\label{eq:numerical_convergence}
\end{figure}
\item \emph{source health tests}: our bounds may be also applied to quantify decreases in the source quality (standards~\cite{Barker2016NISTDS,AIS31} require such countermeasures to be implemented). Namely, our bounds provide the min-entropy rate in the ideal case (independent samples), which upperbounds the actual rate. In practical designs one prefers sources with small entropy deficiency  \cite{Hennebert2013,PawlowskiJO15,Bedekar2015}, where our results offer 
more accurate estimates.
%Our result  shows that entropy in such regimes "mixes" much better (at least for independent distributions) than for general sources.
\end{enumerate}

\subsubsection{Applications to extremal inequalities for Renyi entropy}

As a byproduct of our analysis, we obtain sharp bounds on Renyi entropy for given Shannon entropy, in \emph{one-shot experiments} (as opposed to the previous application).
The motivation comes from cryptographic tasks such as key derivation, which demand "enough" Renyi entropy available. We ask
what can be said about Renyi entropy of a distribution, if only its Shannon entropy is known. Using our techniques, the precise answer can be given for any $\alpha>0$. Below in \Cref{fig:extreme_entropies} we illustrate such a bound for Renyi entropy of order $\alpha= 2$ (denoted by $H_2$). The Python script is included in \Cref{app:python}.
%discuss these applications in a foregoing paper. 
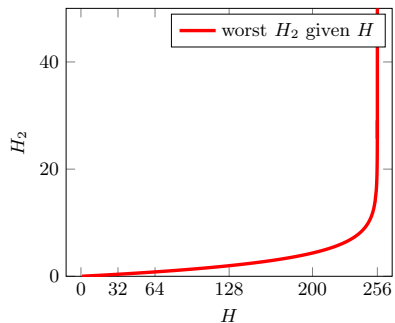
\begin{figure}[!h]
\centering
\begin{tikzpicture}[scale=0.8]
\begin{axis}[
		enlarge x limits = 0.05,
    	title = {},  % whatever name you want
    	ylabel = {$H_2$},
    	xlabel = {$H$},
    	ymin = 0, ymax = 50,
    	xmin = 0, xmax = 256,
    	xtick = {0,32,64,128,200,256},
	legend style = {anchor = north east},
    	legend entries = {worst $H_2$ given $H$}
    	]
    	\addplot[red, ultra thick] table {extreme.csv};
    	\draw[red, ultra thick] (256,257)--(256,500);
    \end{axis}
\end{tikzpicture}
\caption{Smallest values for $H_2(X)$ when $H(X)$ is fixed (minimized over the choice $X$), for distributions over 256 bits.
%The points were calculated by solving 
%$H_1(X) = k$ where $X$ as as in \Cref{thm:maximize_solution} and $k$ changes
%from $0$ to $256$. 
Even with $H(X) = 255.999$ we could have only $H_2(X)=35.7$, so the convergence to full entropy is extremely slow. 
}
\label{fig:extreme_entropies}
\end{figure}

\subsubsection{Our results in detail}

%Below we compare our results with
%\cite{HolensteinRenner2011} (see \Cref{tab:comparison}) and elaborate more on %the result and techniques.
\paragraph{Optimal bounds on exponential moments of surprises}

We compute the maximal value of the moment generating function of the surprise, when the distribution is over the finite space $\mathcal{X}$ and has a certain ammount of entropy $k$. 
Essentialy, we need to solve the following optimization task\footnote{This program is equivalent to what we announced in the introduction. 
}
\begin{align}\label{eq:maximize_MGF}
\begin{aligned}
\underset{X}{\mathrm{maximize}} && \mathbb{E}_{x\sim X}\exp\left(tH(X)-t\log\frac{1}{P_X(x)} \right) \\
\mathrm{s.t.} && H(X) = k.
\end{aligned}
\end{align}
over random variables $X$ taking values in $\mathcal{X}$,
where $t$ is a parameter.
Our main result characterizes the optimal solution to \eqref{eq:maximize_MGF}.
\begin{theorem}[Sharp surprise exponential moments]\label{thm:maximize_solution}
For any $t\geqslant -1$, the optimal solution to \eqref{eq:maximize_MGF} is given by
%Before discussing further applications, we highlight some intuitions
%Before discussing further applications, we highlight some intuitions
\begin{align}\label{eq:optimal_distribution}
P_X(x) = \left\{
\begin{array}{r}
\theta,  \ x=x_0 \\
\frac{1-\theta}{|\mathcal{X}|-1}, \ x\not=x_0\\
\end{array}
\right.
\end{align}
for some $\theta \in \left(\frac{1}{|\mathcal{X}|},1\right)$ and $x_0 \in \mathcal{X}$.
%Before discussing further applications, we highlight some intuitions
%Before discussing further applications, we highlight some intuitions
%Before discussing further applications, we highlight some intuitions
%Before discussing further applications, we highlight some intuitions
\end{theorem}
%Before discussing further applications, we highlight some intuitions
%for the result and prove techniques.
%Before discussing further applications, we highlight some intuitions
%Before discussing further applications, we highlight some intuitions
%Before discussing further applications, we highlight some intuitions
\begin{remark}[Intuitions]
The result essentially says that the optimal distribution is a combination of a unit mass and a uniform distribution.
This is a consequence that the optimization program in \eqref{eq:maximize_MGF} basically exhibits two different behaviors: 
concavity for small probability weights and convexity for larger probability weights. While this charecterization is simple the proof is not, as  the standard convex/concave programming framework cannot be applied.
\end{remark}

\begin{remark}[Techniques]
To handle constrained optimization like \eqref{eq:maximize_MGF},
the standard approach is to skip the constraint,
adding instead a corresponding penalty term to the objective (forming the so called Lagrangian). In our case we get
$$L = \sum_{i}p_i\exp\left(tk-t\log\frac{1}{p_i}\right) + \lambda\left(H(p)-k\right)$$ for a weight $\lambda$, to be maximized over probability vectors $p$ (the dual problem). 

By the elegant methods of \emph{majorization theory} we show that
the dual problem is solved by a distribution as in \Cref{eq:optimal_distribution}.
Basically this is because $L$ as a function of $p$ is convex
when restricted to variables $p_i > c$ and concave when restricted to variables $p_i < c$, where $c$ is some constant. To maximize in the convex region the best choice is to have only one $i$ such that $p_i > c$.
To maximize in the concave region the optimal choice is
$p_i = p_j < c$ when $p_i,p_j < c$ (see \Cref{fig:optimization} below).

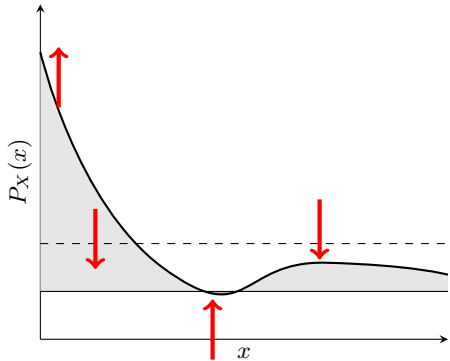
\begin{figure}[h!]

\begin{tikzpicture}
\begin{axis}[
  axis lines = left,
  axis line style= {-stealth},
  ticks = none,
  clip=false,
  xmin = 0,
  xmax = 5,
  ymin=0,
  ymax = 1.0,
  enlarge y limits = 0.2,
  xticklabels=\empty,
  yticklabels=\empty,
%  extra x ticks = {0,0.5,1,1.5,2,2.5,3,3.5,4,4.5},
  xlabel = $x$,
  ylabel = $P_X(x)$,
  legend pos=north east%,
%  every axis y label/.style={
%    at={(ticklabel* cs:1.05)},
%    anchor=south,
%}
]
\addplot+[black, thick, name path =A, smooth, tension=1.1, samples = 50, mark = none] coordinates {(0,1) (1.5,0.1) (3.5,0.12) (5,0.07)} 
node[pos=0.05, above, yshift = -0.55cm] (upa) {}
node[pos=0.15, below, yshift = -0.55cm] (downa) {}
node[pos=0.45, below, yshift = -0.2cm] (upb) {}
node[pos=0.7, above, yshift = -0.2cm] (downb) {};
\addplot[black, dashed, mark = none, domain=0:5] {0.2};
\addplot[name path = C, domain=0:5] {0.0};
\node[above = 0.8cm of upa] (upaa) {};
\draw[->, red, ultra thick] (upa) -- (upaa);
\node[below = 0.8cm of downa] (downaa) {};
\draw[->, red, ultra thick] (downa) -- (downaa);
\node[below = 0.8cm of upb] (upbb) {};
\draw[->, red, ultra thick] (upbb) -- (upb);
\node[above = 0.8cm of downb] (downbb) {};
\draw[->, red, ultra thick] (downbb) -- (downb);
 \addplot fill between[ 
    of = C and A, 
    soft clip={domain=0:5},
    split, % calculate segments
    every even segment/.style = {gray!20},
    every odd segment/.style  = {gray!60}
  ];
\end{axis}
\end{tikzpicture}
\caption{Pull-away in the convex region, squeeze in the concave region.}
\label{fig:optimization}
\end{figure}

Unfortunatelly, our $L$ fails to satisfy desired 
convexity/concavity properties and the solution to the dual problem is not guaranteed to be optimal for the original problem. To rule out this possibility, called the \emph{duality gap}, we use second order conditions.
These conditions are basically a variational analysis of how $L$ changes when a stationary point $p$ moves a little bit in a way consistent with the constraints (to make this precise, the constraints are linearized and
$L$ is approximated up to second order terms). In our case we conclude that $p$ must be anyway as in \eqref{eq:maximize_MGF}, which completes the proof.

The argument is technically more complicated than sketched here, as 
we need to make sure that the optimal $p$ has no zero weights (by restricting to positive weights we avoid singularities in anaylysis).
We illustrate our proof in \Cref{fig:proof_map} below. The details are given in \Cref{sec:maximize_MGF_constrained}.
\end{remark}

\begin{figure}[t!]
\resizebox{0.99\textwidth}{!}{
\begin{tikzpicture}
\node[draw, rectangle, align=center] (a1) at (0,0) {reduce to maximization of Renyi entropy of order $1+t$ \\
(given Shannon entropy constraints)};
\node[draw, rectangle, align=center, below = 1cm of a1] (b1) {consider the Lagrangian $L$\\
(makes the problem unconstrained)};
\node[draw, rectangle, align=center, below = 1cm of b1] (c1) {find stationary points of $L$ \\
\textbf{result}: weights of $P_X$ take 2 non-zero values \\
\textbf{tool}: first-order optimality conditions
};
\node[below = 0.75cm of c1] (d3) {second order conditions};
\node[draw, rectangle, align=center, below left = 0.5cm and 1cm of d3
] (d1) { $L$ has negative curvature: \\
\textbf{result}: $P_X$ combines a point and a flat dist. \\
\textbf{tool}: majorization theory (Schur convexity)
};
\node[draw, rectangle, align=center, below right = 0.5cm and 1cm of d3] (d2)  {
other curvatures: \\
\textbf{result}: $P_X$ combines a point and a flat dist. \\
\textbf{tool}: second-order optimality conditions
};
\node[draw, rectangle, align=center, below = 3.5cm of d3] (e1) {
\textbf{summary}: $P_X$ combines a point and a flat distribution \\
(may still have zero weights)};
\node[draw, rectangle, align=center, below = 1cm of e1] (f1) {
\textbf{result}: zero-weights are suboptimal \\
\textbf{tool}: implicit function theorem 
};
\draw[->] (a1)  -- (b1) node[midway, right, align=center]{penalty terms};
\draw[->] (b1)  -- (c1) node[midway, right, align=center]{first derivative test};
\draw[->] (e1)  -- (f1) node[midway, right, align=center]{
parameterize by the support
};
\draw[->] (d1) -| (e1) {};
\draw[->] (d2) -| (e1) node [below right = 0.4cm and 3cm of d1]{case analysis} {}; 
\draw[-] (c1) -- (d3);
\draw[->] (d3) -| (d1) {};
\draw[->] (d3) -| (d2) {};
\end{tikzpicture}
}\caption{The overview of the proof of \Cref{thm:maximize_solution}}
\label{fig:proof_map}
\end{figure}
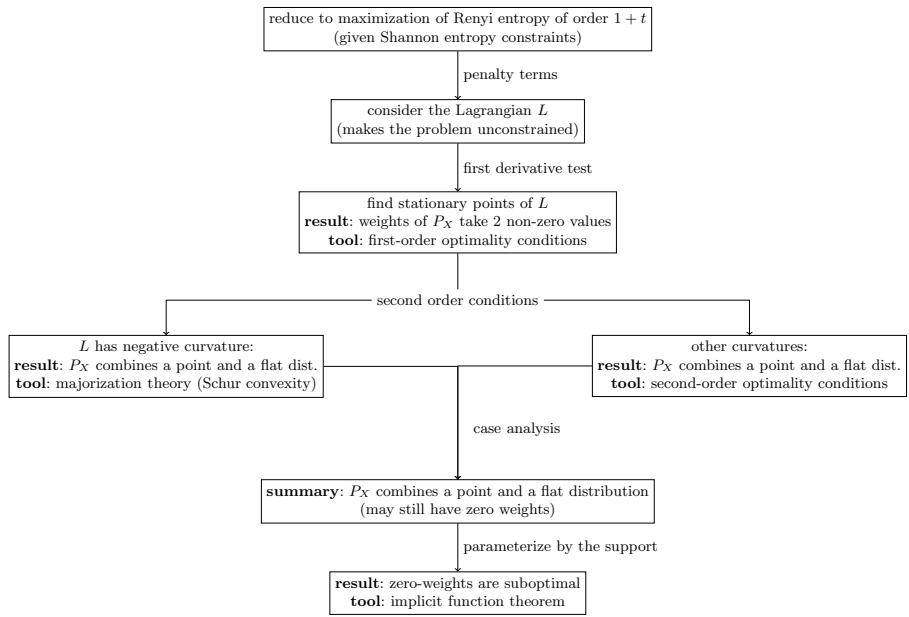

\subsubsection{Optimal sub-exponential tails under entropy constaints}

From \Cref{eq:optimal_distribution} we derive best possible bounds
on the surprise exponential moments. 
There are two technical difficulties to handle.
The first issue is that the distribution in \eqref{eq:optimal_distribution} is given in terms of the bias $\gamma = \theta-\frac{1}{|\mathcal{X}|}$, whereas we need to parameterize it in terms of the entropy amount $k$. Resolving the equation
$H(X) = k$ involves inverting non-elementary equations of the form $\theta\log\theta = c$, and can be done with the use of the Lambert-W function. The second issue is that plugging 
\eqref{eq:optimal_distribution} into \eqref{eq:maximize_MGF} does not lead to a clean formula, and requires some calculus to get clear bounds.

\begin{corollary}\label{cor:sub_exp_tails_improved}
For independent sequence $X_1,X_2,\ldots,X_n$ of random variables over 
$\mathcal{X}$ with entropy defficiency $\Delta_i = H_0(X_i) - H(X_i)$ at most $\Delta$, we have 
\eqref{eq:convergence_2} with the following parameters (below $N=\log|\mathcal{X}|$)
\begin{enumerate}[(a)]
\item for $\Delta = \omega(N^{-1})$ the tail is
\begin{align*}
\epsilon =\left\{
\begin{array}{rl}
 \exp\left(-\frac{t^2}{2 n \Delta\log(N\Delta) }\right), & \quad t < n\Delta \\
 \exp\left(-\frac{t}{2 \log(N\Delta)}\right), & \quad t > n\Delta
\end{array}\right.
\end{align*}
\item when $\Delta = O(N^{-1})$ the tail is
\begin{align*}
\epsilon = \left\{
\begin{array}{rl}
 \exp\left(-\frac{t^2}{2 n \Delta }\right),  & \quad t < n\Delta \\
 \exp\left(-\frac{t}{2}\right), & \quad t > n\Delta
\end{array}
\right.
\end{align*}
\end{enumerate}
\end{corollary}
The bounds in \Cref{tab:comparison} follow by setting $t = n\delta$ and by combining the formulas for $t<n\Delta$ and $t>n\Delta$ into one, with the maximum function.
\begin{remark}{Examples}
To illustrate the bounds, consider the minimal number $n$ to have
\eqref{eq:convergence_2}, for fixed accuracy $\delta$ and error probability $\epsilon)$-approximation
\begin{enumerate}[(a)]
\item when $\Delta_i = O(1)$ (constant defficiency) we get $n = O(1) \log|\mathcal{X}|\delta^{-2}\log(1/\epsilon)$, saving a factor of 
$\Omega(\log|\mathcal{X}|)$ comparing to~\cite{HolensteinRenner2011}
\item when $\Delta_i = O\left(\frac{1}{H_0(X_i)}\right)$ then 
$n = O(1) \delta^{-2}\log(1/\epsilon)$, which doesn't depend on the alphabet anymore. 
\end{enumerate}
\end{remark}
The proof appears in \Cref{sec:improved_subexp}. 

\subsection{Organization}

We provide background definitions and auxiliary fact used in 
\Cref{sec:prelim}. The proofs are given in \Cref{sec:proofs}.
We conclude our work in \Cref{sec:conc}

\section{Preliminaries}\label{sec:prelim}

We say that the vector $p=(p_i)_i$ is a probability vector
if its entries are non-negative and add up to 1.
The distribution of a radom variable $X$ is denoted by $P_X(x) = \Pr[X=x]$.
The surprise of $X$ is a random variable
$x\rightarrow \log\frac{1}{P_X(x)}$. The Shannon entropy is the expected surprise $\mathbb{E}_{x\sim X}\log\frac{1}{P_X(x)}  =\sum_{x}P_X(x)\log\frac{1}{P_X(x)}$. The Renyi entropy of order $\alpha$ is defined as $-\frac{1}{\alpha-1}\log\sum_{x}P_X(x)^{\alpha}$.

\subsection{Sub-Gaussian Random Variables}

Below we remind basic facts from the theory of subgaussian and subexponential distributions (we refer to~\cite{Vershynin2016} for a detailed threatment).
\begin{definition}[Sub-gaussian tails]
A real-valued random variable $X$ with mean $\mu$ is sub-gaussian with parameter $\sigma^2$ if for all real $t$ we have
\begin{align*}
\mathbb{E}\exp\left(t (X-\mu)\right) \leqslant \exp\left(\frac{\sigma^2 t^2}{2}\right).
\end{align*}
\end{definition}

\begin{lemma}[Independent sub-gaussian random variables]
For independent $X_i$ each sub-gaussian with parameters $\sigma_i^2$, the sum $X=\sum_{i}X_i$ is sub-gaussian
with parameter $\sum_{i}\sigma_i^2$.
\end{lemma}

\subsection{Sub-Exponential Random Variables}

\begin{definition}[Sub-exponential random variables]
A real-valued random variable $X$ with mean $\mu$ is sub-exponential with parameters $(\sigma^2,b)$ if for all $|t|\leqslant \frac{1}{b}$ we have
\begin{align*}
\mathbb{E}\exp\left(t (X-\mu)\right) \leqslant \exp\left(\frac{\sigma^2 t^2}{2}\right).
\end{align*}
\end{definition}

\begin{lemma}[Sub-exponential tails]
For $X$ as above we have 
\begin{align*}
\Pr[X>\mu  +t]\leqslant 
\left\{ 
\begin{array}{rl}
\exp\left( -\frac{t^2}{2\sigma^2}\right), & 0\leqslant t \leqslant \frac{\sigma^2}{b} \\
\exp\left( -\frac{t}{b}\right), & \frac{\sigma^2}{b} < t
\end{array}
\right.
\end{align*} 
\end{lemma}

\begin{lemma}[Independent sub-exponential random variables]\label{lemma:sub_exp_sum}
For independent $X_i$ each sub-exponential with parameters $(\sigma_i^2,b_i)$, the sum $X=\sum_{i}X_i$ is sub-exponential
with parameters $(\sum_{i}\sigma_i^2,\max_{i} b_i)$.
\end{lemma}

\subsection{Optimization theory}
In this section we very briefly remind some concepts from optimization theory (see for example \cite{Floudas2006} for a reference). Suppose that we want to solve a problem of the form
\begin{align*}
\begin{aligned}
\underset{p \in D}{\mathrm{maximize}} && f(p) \\
\mathrm{s.t.} &&
\begin{array}{l}
h_i = 0,\quad i\in I \\
g_j \geqslant 0,\quad j\in J.
\end{array}
\end{aligned}
\end{align*}
where $D\subset \mathbb{R}^d$ is an open set. Any point $p$ satisfying the constraints is called feasible. We call the maximizer $p=p^{*}$ the optimal point. The inequality constraint $g_j$ is called active at $p$ if $g_j(p) = 0$. All equality constraints are active at any feasible point.

The optimal point can be characterized by the so called KKT conditions, provided that 
certain regularity properties are satisfied.

\begin{definition}[LICQ constraint qualification]
We say that the LICQ constraint qualification holds, if at the optimal point
the gradients of the active constraints are linearly independent.
\end{definition}
If the LICQ condition is satisfied, then the optimal point $p=p^{*}$ is a stationary point to the Lagrangian formulated as
\begin{align*}
L = f + \sum_{i}\lambda_i \cdot  h_i + \sum_{j}\cdot \mu_j g_j
\end{align*}
where $\lambda_i \in \mathbb{R}$ for $i\in I$ and $\mu_j \geqslant 0$ for $j\in J$ are some weights (non-zero only for active constraints). This leads to the so called \emph{first order conditions}
\begin{align*}
\frac{\partial L}{\partial p}(p^{*}) = 0.
\end{align*}
In lack of convexity properties, could be that $p^{*}$ is only stationary
to $L$, but is not optimal for $L$. Still, the following second order conditions are satisfied
\begin{align*}
d^T \cdot \frac{\partial^2 L}{\partial p^2}(p^{*}) \cdot d\geqslant 0
\end{align*}
for all vectors $d$ such that $\frac{\partial g_j}{\partial p}(p^{*})\cdot d = 0$ for active $j\in J$ and $\frac{\partial h_i}{\partial p}(p^{*})\cdot d = 0$ for $i\in I$. These vectors are called "tangent" because they discribe small perturbation of the point that are consisent with the constraints.

\subsection{Majorization theory}

\begin{definition}[Vectors majorization]
For two vectors $u,v\in\mathbb{R}^d$ we say that 
$u$ majorizes $v$, and denote by $u\succ v$,
 if the following inequalities
\begin{align*}
\sum_{i=1}^{j} u'_{i} \geqslant \sum_{i=1}^{j} v'_{i} \quad
\text{ for } j=1,\ldots,d
\end{align*}
where $u'$ and $v'$ are vectors with the same components as
$u$ and $v$ respectively, sorted in the non-decreasing order.
\end{definition}

\begin{definition}[Schur convexity]
We say that $f:\mathbb{R}^d\rightarrow \mathbb{R}$ is Schur-convex (abbreviated to S-convex) whenever
$f(u)\geqslant f(v)$ for any $u,v\in\mathbb{R}^d$ such that
$u$ majorizes $v$.
\end{definition}

\begin{proposition}[S-Convexity Criteria~\cite{marshall11}]\label{lemma:sconvex_criteria}
The following statements are true
\begin{itemize}
\item Every symmetric and convex function is S-convex.
\item If $f:\mathbb{R}^d\rightarrow \mathbb{R}$ is increasing in each coordinate
and $h_i$ for $i=1,\ldots,d$ are S-convex then the composition
$f(h_1,\ldots,h_d)$ is S-convex.
\end{itemize}
\end{proposition}

\section{Main Result}\label{sec:proofs}

\subsubsection{Optimizing the surprise MGF under entropy constraints}
\label{sec:maximize_MGF_constrained}

%\subsection{Maximizing the MGF given entropy constraints}

\begin{proof}[Proof of \Cref{thm:maximize_solution}]
We consider \eqref{eq:maximize_MGF} for some alphabet $\mathcal{X}$ of fixed size $N$. We will prove that for some constants $\alpha>\beta>0$, the distribution $P_X$ optimal to \eqref{eq:maximize_MGF}
satisfies
\begin{align}\label{eq:maximize_shape_0}
\forall x: \ P_X(x) \in \{\alpha,\beta\},  \ \#\{x: P_X(x)=\alpha \} =1
\end{align}
Note that \eqref{eq:maximize_MGF} is equivalent to
\begin{align}\label{eq:maximize_0}
\begin{aligned}
\underset{p}{\mathrm{maximize}} && \sum_{i}p_i\mathrm{e}^{t H(p)-t\log\frac{1}{p_i}} \\
\mathrm{s.t.} && 
\left.
\begin{array}{r}
H(p) = k
\end{array}
\right.
\end{aligned}
\end{align}
over \emph{probability vectors} $p\in\mathbb{R}^N$. Note also that
we have to assume that $t \geqslant -1$, as for $t < -1$
the value of \eqref{eq:maximize_0} is unbounded (because the objective equals $\mathrm{e}^{H(p)}\sum_i p_i^{1+t}$ which is arbitrarily big whenever one of $p_i$ is close to zero). For $t=-1$ we see that the objective
in \eqref{eq:maximize_0} is constant and our statement is trivialy true. 
Thus we can assume that $t>-1$.
\begin{claim}[The optimal solution is not flat]
If $k < \log N$ then the solution $p=p^{*}$ to \eqref{eq:maximize_2} has at least two different non-zero entries.
\end{claim}
\begin{proof}[Proof of Claim]
If the optimal solution is a flat distribution, then $H(p^{*}) = k$ implies
$p^{*}_i = 2^{-k}$ for all $i$ such that $p^{*}_i > 0$\footnote{In particular $2^k$ must be integer, but we don't use this observation}. This means that the value of \eqref{eq:maximize_0}
is $1$. However, by the Jensen inequality applied to a strictly convex function $u\rightarrow \mathrm{e}^{u}$, for any probability vector $q$ 
such that $q_i \not q_j \not=0$ for some $i,j$ we have 
\begin{align*}
\sum_{i}q_i\mathrm{e}^{t H(q)-t\log\frac{1}{q_i}} >
\mathrm{e}^{t\sum_i q_i \left(H(q) - \log\frac{1}{q_i} \right) } = 
\mathrm{e}^{t H(q) - tH(q)} = 1.
\end{align*}
In other words, the objective is strictly bigger than $1$ for any non-flat distribution $q$. We conclude that all feasible distributions must be flat.

Now, if $k < \log N$ we consider the probability vector $p$ given by
$p_1 = \delta$, $p_i = \frac{1-\delta}{N-1}$ for $i\not = 1$, and solve
the equation $H(p) = k$. This is equivalent to
\begin{align*}
\delta\log\frac{1}{\delta} + (1-\delta)\log\frac{N-1}{1-\delta} = k
\end{align*}
and because the left-hand side equals $\log N$ when $\delta= N^{-1}$ and $0$ when $\delta = 1$, there is a solution
$\frac{1}{N} < \delta < 1$ (by continuity and the intermediate value theorem). This solution satisfies $0<p_1 < p_2$ hence is not flat, a contraddiction. \qed
\end{proof}
Since for all feasible points $H(p)$ is constant, the objective in
\eqref{eq:maximize_0} can be simplified and the solution is the same as for the program
\begin{align}\label{eq:maximize_2}
\begin{aligned}
\underset{p}{\mathrm{maximize}} &&  \sum_{i} p_i^{1+t} \\
\mathrm{s.t.} && 
\left.
\begin{array}{r}
H(p) = k
\end{array}
\right.
\end{aligned}
\end{align}
over probability vectors $p \in \mathbb{R}^N$. 
In the first step we prove a somewhat weaker result, namely that the optimal solution $p=p^{*}$ for some $0<\alpha< \beta$ satisfies
\begin{align}\label{eq:maximize_shape}
\forall i:\ p^{*}_i \in \{0,\alpha,\beta\},  \ \#\{i: p^{*}_i=\alpha \} =1
\end{align}
Note that we can skipp the zero entries of $p$,
as this doesn't change other constraints $H(p) = k$ and $\sum_{i}p_i = 1$
neither the objective function. More precisely, 
if $p^{*}$ is optimal for \eqref{eq:maximize_2} and $I = \{i: p^{*}_i > 0 \}$ then $ p = p^{*}_{i\in I}$ is a local maximizer of 
\begin{align}\label{eq:maximize_3}
\begin{aligned}
\underset{p}{\mathrm{maximize}} &&  \sum_{i\in I} p_i^{1+t} \\
\mathrm{s.t.} && 
\left\{
\begin{array}{r}
H(p) = k \\
\sum_{i\in I} p_i = 1 \\
\forall i\in I:\quad p_i > 0
\end{array}
\right.
\end{aligned}
\end{align}
over vectors $p\in\mathbb{R}^{|I|}$.
\begin{claim}[Regularity conditions hold]
If $p=p^{*}$ is optimal to \eqref{eq:maximize_3}, then
the LICQ condition is satisfied at $p^{*}$.
\end{claim}
\begin{proof}[Proof of Claim]
The active constraints are $\sum_{i}p_i = 1$ and  
$H(p) = k$. We have $\frac{\partial \sum_{i}p_i}{\partial p_i} = 1$,
and $\frac{\partial H(p)}{\partial p_i} = \log\frac{1}{p_i}-1$. If the gradients are linearly dependent at $p^{*}$ then $\lambda_1 + \lambda_2\left(\log\frac{1}{p^{*}_i}-1 \right)=0$ for all $i \in I$. If $\lambda_2=0$ then $\lambda_1 = 0$ and we are done. If $\lambda_2\not=0$ then for some constants $c$ and all $i$ we have $p^{*}_i = c$, a contraddiction to the previous claim that $p^{*}$ cannot be flat.
\qed
\end{proof}
The lagrangian associated with \eqref{eq:maximize_3} is
\begin{align}\label{eq:lagrange}
L(p,\lambda_1,\lambda_2) = \sum_{i\in I} p_i^{1+t} - \lambda_1 \left(\sum_{i\in I}p_i\log\frac{1}{p_i}-k\right) - \lambda_2\left(\sum_{i\in I}p_i - 1 \right).
\end{align}
By the first order conditions applied to \eqref{eq:lagrange}
we (partially) characterize the optimal solution.

\begin{claim}[The solution combines two flat distributions]
If $p=p^{*}$ 
solves \eqref{eq:maximize_3}, then 
its entries take only two values: 
$p^{*}_i \in \{\alpha,\beta\}$ for some $0 <\alpha<\beta$ and all $i\in I$.
\end{claim}
\begin{proof}[Proof of Claim]
By the first order conditions (justified because of the regularity proven in the last claim) for some $\lambda_1,\lambda_2$ we have
\begin{align}\label{eq:foc}
\forall i\in I:\quad  0 = \frac{\partial L}{\partial p_i}(p^{*})  = (1+t)(p^{*})_i^{t} - \lambda_1\left(\log\frac{1}{p^{*}_i}-1 \right) - \lambda_2
\end{align}
Note that \Cref{eq:foc} is equivalent to $g(p^{*}_i) = c$ for all $i\in I$ and some constant $c$, where $g$ is the function defined by $g(u) = (1+t)u^t - \lambda_1\log\frac{1}{u}$.
The derivative equals $\frac{\partial g}{\partial u} = t(1+t)u^{t-1}+\frac{\lambda_1}{u} $, and changes its sign at most once over $u\in (0,1)$. Thus any equation $g(u) = c$ has at most two solutions in $u\in (0,1)$. In particular, since $g(p^{*}_i) = c$ for all $i$,
there are two values $\alpha,\beta$ possible for $p^{*}_i$ where $i\in I$.
Note that $0\not=\alpha\not=\beta$ as we proved that $p^{*}$ is not  uniform. We can assume $\alpha<\beta$.
%there are numbers $\alpha,\beta$ such that $p_i \in \{\alpha,\beta\}$.
\qed
\end{proof}

The last claim is a step towards the characterization \Cref{thm:maximize_solution}, but we need to establish that the weight
$\alpha$ is used only once and that no zero-weights occur. 
The numbers $\alpha,\beta$ from the last claim may depend only $t$ and $k$.
Below we show that the dependency is basically limited to $k$.
\begin{claim}[Optimal probabilities don't depend on $t$]
For any fixed $k$, there exist finitely many choices for the optimal solution (the choices don't depend on $t$).
\end{claim}
\begin{proof}[Proof of Claim]
By last claim, we have 
\begin{align}\label{eq:mixture_of_two}
\left.
\begin{array}{r}
N_{\alpha}\alpha + N_{\beta}\beta = 1 
\\
N_{\alpha}\alpha\log\frac{1}{\alpha}+N_{\beta}\beta\log\frac{1}{\beta} = k
\end{array}
\right.
\end{align}
for some natural numbers $N_{\alpha} = \#\{i:\ p^{*}_i=\alpha\},
N_{\beta} = \#\{i:p^{*}_i = \beta\}$ such that $N_{\alpha}+N_{\beta} = |I|\leqslant N$. 
Let $g = (g_1,g_2):\mathbb{R}^2 \rightarrow \mathbb{R}^2$ be a function of $\alpha,\beta$ with parameters $N_{\alpha},N_{\beta}$,
such that
$g_1$ and $g_2$ are the left-hand sides of the first and second equation of 
\eqref{eq:mixture_of_two} respectively.
The jacobian of $g$ with respect to $\alpha,\beta$ equals
\begin{align*}
\mathrm{det} \left[ \frac{\partial g}{\partial (\alpha,\beta)} \right]
= \mathrm{det}
\left[
\begin{array}{rl}
N_{\alpha} &\ N_{\alpha}\left(\log\frac{1}{\alpha}-1 \right)\\
N_{\beta} &\ N_{\beta}\left(\log\frac{1}{\beta}-1 \right)
\end{array}
\right] = N_{\alpha}N_{\beta}\left(\log\frac{1}{\beta}-\log\frac{1}{\alpha} \right).
\end{align*}
Since the optimal solution is not flat, we have $N_{\alpha},N_{\beta} \not= 0$ and $\alpha\not=\beta$. It follows that
$\mathrm{det} \frac{\partial g}{\partial (\alpha,\beta)} \not= 0$.
By the implicit function theorem~\cite{krantz2012implicit}, for any $c_1,c_2$
there is at most one solution to $g(\alpha,\beta) = (c_1,c_2)$.
In particular, setting $c_1=1,c_2 = k$ we conclude that \eqref{eq:mixture_of_two} has at most one solution for fixed parameters $N_{\alpha}, N_{\beta}$. There are finitely many choices for these paramters and the claim follows. \qed
\end{proof}
We will argue that it can be assumed that the hessian matrix of $L$ at the optimal point $p^{*}$ is negatively defined. %In our case, the matrix is zero outside of the diagonal.

%In general, $\frac{\partial^2 L}{\partial p^2}$ is negatively defined only towards feasible directions. In our case, however, the constraints and objective are symmetric, and moreover the optimal point is of a special form. the set of feasible directions is reach because there is a lot of symmetries in the optimal solution 
\begin{claim}[The hessian diagonal is non-zero at the optimal point]
For all but finitely many values of $t$, at the optimal point we have $\frac{\partial^2 L}{\partial p_i^2} \not= 0$ for all $i$.
\end{claim}
\begin{proof}[Proof of Claim]
We have
\begin{align}\label{eq:lagrange_hessian}
\frac{\partial^2 L}{\partial p_i p_j} = 
\left\{
\begin{array}{rl}
0, & i\not=j \\
 t(1+t)p_i^{t-1}+\lambda_1 p_i^{-1},  & i=j.
\end{array}
\right.
\end{align}
Therefore if $\frac{\partial^2 L}{\partial p_i^2}(p^{*}) = 0$ for some $i$, then
$t(1+t)(p^{*}_i)^{t-1}+\lambda_1 (p^{*}_i)^{-1} = 0$. Combining the last claim with
\eqref{eq:foc} and assuming, without loss of generality, that $p^{*}_i = \alpha$ ($p^{*}_i = \beta$ is analogues) we obtain
\begin{align}\label{eq:soc_degenrated}
\begin{array}{r}
t(1+t)\alpha^{t-1}+\lambda_1 \alpha^{-1} = 0 \\
(1+t)\alpha^{t} - \lambda_1\left(\log\frac{1}{\alpha}-1 \right) - \lambda_2 = 0\\
(1+t)\beta^{t} - \lambda_1\left(\log\frac{1}{\beta}-1 \right) - \lambda_2 = 0
\end{array}
\end{align}
Computing $\lambda_1$ from the second and third equation and plugging into the first equation we obtain (note that $\alpha\not=\beta$ is guaranteed)
\begin{align*}
t\alpha^t - \frac{\alpha^t-\beta^t}{\log{\alpha}-\log{\beta}} = 0.
\end{align*}
This is equivalent to
$t c_1^t  + c_2 = 0$
where $\mathrm{poly}(t)$ is a polynomial in $t$ and $c_1,c_2$ are constants depending on $\alpha,\beta$. It follows that $t$ takes only finitely many values for fixed $\alpha,\beta$. Since $\alpha,\beta$ take finitely many values, by the last claim, there are finitely 
many numbers $t$ such that \eqref{eq:soc_degenrated} is satisfied by some $\alpha,\beta$.
\qed
\end{proof}
\begin{claim}[The hessian is negative definite,
or the characterization \eqref{eq:maximize_shape} holds.]
For all but finitely many $t$, the optimal point is 
as in \eqref{eq:maximize_shape} or makes the hessian of $L$ negative definite.
\end{claim}
\begin{proof}[Proof of Claim]
%$\frac{\partial^2 L}{\partial p_i^2} \not= 0$. But then, 
Let $t$ be as in the last claim and let $p^{*}$ be optimal for \eqref{eq:maximize_3}. By the second order conditions we have
\begin{align}\label{eq:negative_hessian_on_feasible}
d^{T}\cdot \frac{\partial^2 L}{\partial p^2}(p^{*})\cdot d \leqslant 0
\quad \text{for } d\in\mathbb{R}^{|I|}:\ \sum_{i\in I}\left(\log\frac{1}{p_i}-1\right)d_i = 0,
\sum_{i\in I}p_i d_i = 0
\end{align}
Define  $N_{\alpha} = \#\{i:p_i= \alpha\}$ and $N_{\beta} = \#\{i: p_i = \beta\}$. 
Suppose first that $N_{\alpha},N_{\beta} > 1$. Choosing $d_{i_1} = \pm \delta, d_{i_2} = \mp\delta$ for $i_1\not=i_2$ such that $p_{i_1} = p_{i_2} = \alpha$ or
 $p_{i_1} = p_{i_2} = \beta$ in \eqref{eq:negative_hessian_on_feasible}
 and using \eqref{eq:lagrange_hessian} yields
 $0\geqslant \left(\frac{\partial^2 L}{\partial p_{i_1}^2} (p^{*}) + 
 \frac{\partial^2 L}{\partial p_{i_2}^2} (p^{*})\right)\cdot \delta^2 = 
2\frac{\partial^2 L}{\partial p_{i_1}^2} (p^{*}) \cdot \delta^2 $ for all $\delta > 0$. Therefore 
\begin{align*}
 \forall i\in I:\quad \frac{\partial^2 L}{\partial p_i^2}(p^{*}) \leqslant 0.
\end{align*}
By the assumption on $t$ and the previous claim, this implies that
$L$ is negative definite at $p^{*}$.

Assume now $N_{\alpha}>1$ but 
$N_{\beta}  =1$. As in the previous part we show that
$\frac{\partial^2 L}{\partial p_i^2}(p^{*}) < 0$ for all $i$ such that $p^{*}_i = \alpha$. 
By \Cref{eq:lagrange_hessian} there exists $c$ such that
\begin{align}\label{eq:maximize_convex_threshold}
\forall i\in I:\quad \mathrm{sgn} \frac{\partial L^2}{\partial p_i^2}(p) 
= \left\{
\begin{array}{rl}
1 & p_i > c \\
0 & p_i = c \\
-1 & p_i < c
\end{array}
\right.
\end{align}
where $c$ depends on $t$ and $\lambda_1$. Hence, if $p^{*}_{i_1} =\alpha$ and $p^{*}_{i_2} = \beta$ then
$\frac{\partial^2 L}{\partial p_{i_1}^2}(p^{*})<0$ means $\frac{\partial^2 L}{\partial p_{i_2}^2}(p^{*})<0$. Therefore $L$ is negative definite at $p^{*}$.

Since we have $N_{\alpha},N_{\beta} \geqslant 1$ (as $p^{*}$ is not a flat distribution) the remaining case is $N_{\alpha} = 1$. But this is precisely
\eqref{eq:maximize_shape}. \qed
\end{proof}

\begin{claim}[The negative definite hessian implies  the characterization \eqref{eq:maximize_shape}]
If the hessian of $L$ is negative definite at the optimal point $p^{*}$,
then $p^{*}$ satisfies \eqref{eq:maximize_shape}.
\end{claim}
\begin{proof}[Proof of Claim]
Let $p^{*}$ be optimal. Since $\frac{\partial L}{\partial p}(p^{*}) = 0$
by the first order conditions, $p^{*}$ is a local maximizer of $L$ (with $\lambda_1,\lambda_2$ being fixed parameters).
%Note that we can write $L(p) = \sum_{i\in I}L_i(p_i)$ where
%\begin{align*}
%L_i = p_i^{1+t}-\lambda_1\left(p_i\log\frac{1}{p_i}-%k\right)-\lambda_2\left(p_i -\frac{1}{|I|}\right)
%\end{align*}
Consider $L$ as a function of $p_{J} = (p_{j})_{j\in J}$
for a fixed subset of indices $J$. Let $c$ be as in
\eqref{eq:maximize_convex_threshold}. By
\eqref{eq:maximize_convex_threshold} and \eqref{eq:lagrange_hessian}
 $L$ is convex
for $p\in S^{+} = \cap_{j\in J}\{ p_j > c\}$
and concave for $p\in S^{-} = \cap_{j\in J}\{ p_j < c\}$, where $c$ is a constant (depends on $\lambda_1,t$). 

Let $i_1\not=i_2$ be such that $p^{*}_{i_1}=p^{*}_{i_2} > c$ for some $i_1\not=i_2$.
Take $J = \{i_1,i_2\}$, fix a positive small number $\delta$ and
define $p'_{i_1} = p^{*}_{i_1}+\delta$, $p'_{i_2} = p^{*}_{i_2}-\delta$
and $p'_i = p^{*}_i$ when $i\not\in\{i_1,\i_2\}$. 
Note that $p'$ majorizes $p^{*}$ and $p',p^{*} \in S^{+}$ for $\delta$ sufficiently small. Because $L$ is symmetric in variables $\{p_j\}_{j\in J}$, by Schur convexity
we have $L(p') > L(p^{*})$. This shows that
there is at most one $i$ such that $p^{*}_i > c$. 

Similarily, take $i_1\not=i_2$ such that $p^{*}_{i_1}<p^{*}_{i_2} < c$.
Let $J = \{i_1,i_2\}$, fix a positive small number $\delta$ and
define $p'_{i_1} = p^{*}_{i_1}+\delta$, $p'_{i_2} = p^{*}_{i_2}-\delta$
and $p'_i = p^{*}_i$ when $i\not\in\{i_1,\i_2\}$. 
Note that $p$ is majorized by $p^{*}$ and $p',p^{*} \in S^{+}$ for $\delta$ sufficiently small. Because $L$ is symmetric in variables $\{p_j\}_{j\in J}$, by Schur convexity we have $L(p') > L(p^{*})$.
This shows that $p^{*}_{i_1}=p^{*}_{i_2} $
whenever $p^{*}_{i_1},p^{*}_{i_2} < c $.

In the first part we established that $\{i: p_i = \alpha\}$ for one $i$, the second part implies $p_j = \beta$ for $j\not=i$. This finishes the proof of  the claim. \qed
\end{proof}

The last two claims imply that the solution to \eqref{eq:maximize_2} is characterized by \eqref{eq:maximize_shape} (for all but finitely many $t$). We now show that probability weights are not zero.
\begin{claim}[The optimal point has only positive entries]
If the optimal point $p^{*}$ is as in \eqref{eq:maximize_shape}
then it satisfies \eqref{eq:maximize_shape_0}.
\end{claim}
\begin{proof}[Proof of Claim]
Let $p^{*}\in\mathbb{R}^N$ be as in \eqref{eq:maximize_shape}. Let $v = \#\{ i:p^{*}_i > 0\}$, and $\delta = \alpha$. Then $p_{i_0} = \delta$ and
$p_{i} = \frac{1-\delta}{v}$ for some other $v-1$ values of $i$.
Moreover we have $\delta > \frac{1-\delta}{v}$.
Since $p^{*}$ is not uniform we have $v > 2^k$. Therefore
$(\delta,v)$ solves the following program over $\delta \in (0,1)$ and integers $N\geqslant v>2^{k}$:
\begin{align}\label{eq:full_support}
\begin{aligned}
\underset{\delta,v}{\mathrm{maximize}} &&  \delta^{1+t} + (1-\delta)\left(\frac{1-\delta}{v}\right)^{t} \\
\mathrm{s.t.} && 
\left.
\begin{array}{r}
\delta\log\frac{1}{\delta} + (1-\delta)\log\frac{v}{1-\delta} = k 
\end{array}
\right.
\end{aligned}
\end{align}
Consider this program under the relaxed assumption that $2^k < v \leqslant N$. We show that the maximum is achieved for $v = N$. Indeed
if $2^k < v < N$ then the gradient of the active constraint is
\begin{align*}
\nabla_{u,v} \left( \delta\log\frac{1}{\delta} + (1-\delta)\log\frac{v}{1-\delta} -k\right) = \left( \log\frac{1-\delta}{v \delta},\frac{1-\delta}{v}  \right)
\end{align*}
and hence satisfies the LICQ condition. The first order conditions yield
\begin{align}
\begin{aligned}{r}
(1+t)\delta^{t}-(1+t)(1-\delta)^{t}v^{-t} & = \lambda \cdot \log\frac{1-\delta}{v \delta} \\
t(1-\delta)^{t+1}v^{-t-1} & = \lambda \cdot \frac{1-\delta}{v}.
\end{aligned}
\end{align}
The second equation implies $\lambda > 0$. The left-hand side of the first equation can be rewritten as
$(1+t)\delta^{t}\left(1- \left(\frac{1-\delta}{\delta v}\right)^t\right) $
and, because $t,\delta  > 0$ its sign equals $\mathrm{sgn}\left(1-\frac{1-\delta}{\delta v}\right)$.
In turn the sign of the right-hand side
equals $\mathrm{sgn}(\lambda)\cdot \mathrm{sgn}\left(\frac{1-\delta}{v\delta}-1 \right)$. Note that, because $\lambda>0$,
the signs are opposite
unless $\frac{1-\delta}{v\delta}=1$. this is not possible as by the assumption on $p^{*}$
we have $\delta\not=\frac{1-\delta}{v}$. This shows that
\eqref{eq:full_support} must be maximized at $v = N$, in particular 
$p^{*}_i\not=0$ for all $i$.
\qed
\end{proof}

\end{proof}

\subsection{Improved sub-exponential tails}\label{sec:improved_subexp}

The following lemma parameterizes 
\eqref{eq:optimal_distribution} in terms of the entropy defficiency.

\begin{lemma}[Entropy defficiency as a function of bias]
\label{lemma:defficiency_from_bias}
Let $X$ be as in \Cref{eq:optimal_distribution}.
Then the bias $\gamma = \theta-\frac{1}{|\mathcal{X}|}$ and the entropy defficiency $\Delta = \log|\mathcal{X}|-H(X)$ are related as in \Cref{tab:1}.
\end{lemma}
The proof appears in \Cref{proof:defficiency_from_bias}.

\begin{table}
\centering
\renewcommand*{\arraystretch}{1.3}
\begin{tabular}{|c|c|c|l|}
\hline
bias & support & regime & entropy deficiency \\
\hline
\multirow{3}{*}{$\gamma$} & \multirow{3}{*}{$N$} & $\gamma N = \omega(1)$ & $\Delta = \Theta\left( \gamma\log \gamma N \right)$
\\ \cline{3-4}
& & $\gamma N = \Theta(1)$ & $\Delta = \Theta\left( \gamma \right)$
\\ \cline{3-4}
& & $\gamma N = O(1)$ & $\Delta = \Theta\left( \gamma^2 N \right)$ \\
\hline
\end{tabular}
\caption{Entropy defficiency as a function of bias.}
\label{tab:1}
\end{table}

\begin{proposition}[MGF as the function of bias]
Let $X$ be as in \Cref{eq:optimal_distribution}. Then
\begin{align}\label{eq:suprise_centered}
H(X)-\log\frac{1}{P_X(x)} &=
\left\{
\begin{array}{rl}
 (1-\theta)\log\frac{\theta(N-1)}{1-\theta}, & \quad x=x_0 \\
 -\theta\log\frac{\theta(N-1)}{1-\theta},    & \quad x\not=x_0
\end{array}\right.
\end{align}
\end{proposition}

\begin{lemma}[Sub-exponential tails of the surprise]
Let $X$ be as in \Cref{eq:optimal_distribution}. 
When $\gamma N = \omega(1)$ then the surprise is sub-exponential
with $\sigma^2 = \gamma\log^2 (\gamma N)$ and $b = \log (\gamma N)$. 
When $\gamma N = O(1)$ then the surprise is sub-exponential
with $\sigma^2 = \gamma^2 N$ and $b = 2$. 
\end{lemma}
\begin{proof}
Let $M_j  = \mathbb{E}_{x\sim X}\left( H(X)-\log\frac{1}{P_X(x)} \right)^{j}$. Note that $M_0 = 1$ and $M_1 = 0$.
By the expansion $\exp(u) = \sum_{j=0}^{\infty}\frac{u^j}{j!}$ we obtain
\begin{align*}
 \mathbb{E}_{x\sim X}\exp\left( tH(X)-t\log\frac{1}{P_X(x)} \right)^{j} =
 1+\sum_{j\geqslant 2}\frac{t^j}{j!}\cdot M_j
\end{align*}
By \Cref{eq:suprise_centered} we obtain
\begin{align}\label{eq:suprise_central_moment}
M_j = \theta(1-\theta)\left((1-\theta)^{j-1}-(-\theta)^{j-1} \right)\log^{j}\frac{\theta(N-1)}{1-\theta}
\end{align}
and hence
\begin{align*}
M_j \leqslant \theta(1-\theta)\cdot \log^{j}\left(1+\frac{\theta N-1}{1-\theta}\right).
\end{align*}
Note that we also have $M_j \leqslant 2^j\log^j N$ by the proof of 
\Cref{lemma:tail_no_constraints}.  Now we split our analysis into the following two cases\\
\underline{Case $\gamma N > 2$}. \\
Assume first that $\theta < 1-\frac{2}{N}$.
Let $\theta = \frac{1}{N}+\gamma$.
Then 
$\frac{\theta N-1}{1-\theta} > 2$ and thus
\begin{align*}
M_j \leqslant 2\theta(1-\theta)\log^{j}\left(\frac{\theta N-1}{1-\theta}\right).
\end{align*}
Moreover $\frac{1}{1-\theta}<\theta N -1$ because of $\theta < 1-\frac{2}{N}$. Therefore
\begin{align*}
M_j & \leqslant 4\theta(1-\theta)\log^{j}\left(\theta N-1\right) \\
& = O\left(\gamma\log^j N\gamma \right)
\end{align*}
For $\sigma^2 = \gamma\log^2 (\gamma N)$, $b = \log (\gamma N)$ and
 $|t| \leqslant \frac{1}{b}$ 
we obtain
\begin{align*}
1+\sum_{j\geqslant 2}\frac{t^j}{j!}M_j \leqslant 
\exp\left(O(1) \cdot \sigma^2 t^2 \right).
\end{align*}
which is also valid when $\theta \geqslant 1-\frac{2}{N}$. 
Note that we need $t\geqslant-1$ in \Cref{thm:maximize_solution}, which is automatically satisfied because $b \geqslant 1$.
\\
\underline{Case $\gamma N < 2$}. \\
We have then $\frac{\theta N-1}{1-\theta} = O(N\gamma)$
and by the Taylor expansion $\log(1+u) = O(u)$ valid for $u=O(1)$ we get
\begin{align*}
M_j \leqslant O\left(\frac{1}{N}\cdot (N\gamma)^j\right)
\end{align*}
For $\sigma^2 = \gamma^2 N$, $b = \gamma N$ and
 $|t| \leqslant \frac{1}{b}$.
we obtain
\begin{align*}
1+\sum_{j\geqslant 2}\frac{t^j}{j!}M_j \leqslant 
\exp\left(O(1) \cdot \sigma^2 t^2 \right).
\end{align*}
Note that we need $t\geqslant-1$ in \Cref{thm:maximize_solution},
for this we can assume $b = \max(\gamma N, 1)$.
\end{proof}
Having proved the last lemma, we are ready to derive
\Cref{cor:sub_exp_tails_improved}. 
\begin{proof}{Proof of \Cref{cor:sub_exp_tails_improved}}
We consider two cases \\
\underline{Case $\gamma N = \omega(1)$ } \\
By \Cref{lemma:defficiency_from_bias}, the assumption $\gamma N = \omega(1)$ is equivalent to
 $\Delta = \omega(N^{-1})$. Also,
 $$b = \log(\gamma N) = \Theta\left(\log(N\Delta) - \log\log(N\Delta)\right) = \Theta(\log (N\Delta))$$
and $$\sigma^2 = \gamma \log^2(\gamma N)  = \Theta\left(\Delta\log (N\Delta) \right).$$ 
  By \Cref{lemma:sub_exp_sum}, 
 the sum of $n$ such surprises is subexponential with $n \sigma^2$
 and $b$, hence the tail for $t < \sigma^2/b$ is
$$\exp\left(-\frac{t^2}{2 n \Delta\log(N\Delta) }\right)$$
\underline{Case $\gamma N = O(1)$} \\
By \Cref{lemma:defficiency_from_bias}, the assumption $\gamma N = O(1)$ is equivalent to $\Delta = O(N^{-1})$. Also,
$$ b = \max(1,\gamma N) = O(1) $$
and $\sigma^2 = \gamma^2 N = \Delta$. 
By \Cref{lemma:sub_exp_sum}, 
 the sum of $n$ such surprises is subexponential with $n \sigma^2$
 and $b$, hence the tail for $t < \sigma^2/b$ is
$$\exp\left(-\frac{t^2}{2 n \Delta }\right)$$.
\end{proof}

\section{Conclusion}\label{sec:conc}

We obtained sharp bounds on exponential moments of the surprise when 
the distribution has a certain (fixed) Shannon entropy.
The analysis we did yields a characterization for related extremal problems involving Renyi entropy.

\printbibliography

\appendix

\section{Proof of \Cref{lemma:defficiency_from_bias}}\label{proof:defficiency_from_bias}

\begin{proof}
Consider the equation
\begin{align}\label{eq:entropy_defficiency}
H(X) = H_0(X)-\Delta
\end{align}
For $X$ as in \Cref{eq:optimal_distribution} we obtain
\begin{align*}
-\delta\log \delta - (1-\delta)\log\frac{1-\delta}{N-1} = 
\log N - \Delta
\end{align*}
which is equivalent to
\begin{align}\label{eq:1}
\Delta = \delta\log \frac{(N-1)\delta}{1-\delta} + 
\log\frac{N(1-\delta)}{N-1}
\end{align}
Introducing $\delta = \frac{1}{N}+\gamma$, we may rewrite it as
\begin{align*}
\Delta = \left(\frac{1}{N}+\gamma\right)\log\left(1+\frac{\gamma N}{1-\frac{1}{N}-\gamma} \right) + \log\left(1-\frac{N\gamma}{N-1} \right).
\end{align*}

\noindent \underline{Case 1: $\gamma N = O(1)$.}
By the Taylor expansion $\log(1+u) = u+O(u^2)$ for $u\leqslant 1$ we obtain
\begin{align*}
\Delta & = \left(\frac{1}{N}+\gamma\right)\left( 
\frac{\gamma N}{1-\frac{1}{N}-\gamma} + O(\gamma^2N^2)
\right) - \frac{N}{N-1}\gamma + O(\gamma^2) \\
& =  \left(\frac{1}{N}+\gamma\right)\left( 
\frac{\gamma N}{1-\frac{1}{N}} + O(\gamma^2N^2)
\right) - \frac{N}{N-1}\gamma + O(\gamma^2) \\
& = O(\gamma^2 N) 
\end{align*}
where in the last line we have used the fact that $\gamma = O(1/N)$.

\noindent \underline{Case 2: $\gamma N = \omega(1)$.}
Multiplying both sides of \Cref{eq:1} by $N$, and using the assumption we obtain
\begin{align*}
N\Delta = N\gamma \log N\gamma + o(N\gamma)
\end{align*}
therefore
\begin{align*}
\Delta = \Theta\left(\gamma\log N\gamma \right). 
\end{align*}
This finishes the proof \qed.
\end{proof}

\section{Codes}\label{app:python}

\begin{lstlisting}[language=Python]
from scipy.optimize import bisect, newton
from math import log
# parameters
key_length = 256
N = pow(2,key_length)
# entropy formulas
def shann_entropy(y):
  return y*log(1/y,2)+(1-y)*log((N-1)/(1-y),2)
def renyi_entropy(y):
  return -log( y*pow(y,1)+(1-y)*pow((1-y)/(N-1),1), 2)
# generating data
with open('extreme.csv','w') as out:
  out.write("x y\n")
  # increement = 1
  for i in range(1,key_length):
    def shann_entropy_eq(y):
      return shann_entropy(y)-i
    y = newton(shann_entropy_eq,0.5)
    out.write("%f %f\n" %  (i,renyi_entropy(y)))
  # more dense sampling when close to full entropy
  for i in range(1,100):
    def shann_entropy_eq(y):
      return shann_entropy(y)-(key_length-1)-i*1.0/100
    y = newton(shann_entropy_eq,0.5)
    out.write("%f %f\n" %  ((key_length-1)+i*1.0/100,renyi_entropy(y)))







\end{lstlisting}

\end{document}